\documentclass[12pt]{article}
\usepackage{amsmath}
\usepackage{amsfonts}
\usepackage{amssymb}
\usepackage{amsthm}
\usepackage{graphicx}
\usepackage{float}
\usepackage{tikz} 
\usepackage{xcolor}
\usepackage{natbib}
\usepackage{bm}
\usepackage{transparent}
\usepackage[pagewise]{lineno}
\usepackage{cancel}
\usepackage[affil-it]{authblk}

\newtheorem{defi}{Definition}
\newtheorem{teo}{Theorem}

\newtheorem{propi}[teo]{Proposition}

\newtheorem{rem}{Remark}
\newcommand{\B}{\boldsymbol}
\newcommand{\C}{\mathcal{C}}
\title{\bf Expert Classification Aggregation}
\author[1,2]{Federico Fioravanti\footnote{f.fioravanti@uva.nl\\ I am grateful to Fernando Tohm\'e, Ulle Endriss, Bernardo Moreno, Agustín Bonifacio, participants of the COMSOC seminar at the ILLC, and two anonymous reviewers for comments and suggestions that led to an improvement of the paper.}}
\affil[1]{GATE, Saint-Etienne School of Economics, Jean Monnet University, Saint-Etienne, France}
\affil[2]{Institute for Logic, Language and Computation, University of Amsterdam, Amsterdam, The Netherlands}
\date{\vspace{-5ex}}
\begin{document}

\maketitle
\begin{abstract}
We consider the problem where a set of individuals has to classify $m$ objects into $p$ categories by aggregating the individual classifications, and no category can be left empty.
An aggregator satisfies \emph{Expertise} if individuals are decisive either over the classification of a given object, or the classification into a given category. 
We show that requiring an aggregator to satisfy \emph{Expertise} (or variants of it) and be either unanimous or independent leads to numerous impossibility results.
%We show how demanding this requirement is, leading to many impossibility results.
%We show that many impossibility  if $m\geq p\geq 2$, there are no liberal and unanimous aggregators.
%We also show, by using a stronger notion of liberalism, that there is no liberal and independent aggregator for $p\leq m\leq p+1$.



\end{abstract}
\section{Introduction}
The problem where a set of $n$ individuals has to classify a set of $m$ objects into $p$ different categories, and no categories should be left empty, can be seen as a relevant one in many scenarios.
For example, consider a set of managers that has to assign a set of workers to different tasks.
One of the natural considerations at hand could be that no task should be left unassigned.
Another example is training neural networks to classify images of cats, dogs, and rabbits, using a large database with an equal proportion of images of each animal. 
Requiring each neural network to classify at least one image of each animal serves as a sanity check for the training process.

\citet{maniquet2016theorem} propose a formal setting to consider this problem, extending \citeauthor{kasher1997question}'s \citeyearpar{kasher1997question} analysis of the Group Identification Problem, the situation where a set of individuals has to classify a subset of them into two categories.\footnote{A similar scenario occurs when a group of experts must select one among themselves to receive an award.
See, for example, \citet{holzman2013impartialnominations,tamura2014impartial}, and \citet{TAMURA201641}.}
Inspired by \citeauthor{Arrow1951-ARRIVA}'s \citeyearpar{Arrow1951-ARRIVA} axioms, Pareto and Independence of Irrelevant Alternatives, they show that an aggregator that satisfies Unanimity, which indicates that unanimous classifications should be respected, and Independence, which requires considering each object separately, must be dictatorial.
Later on, \citet{cailloux2023classification} weaken the Unanimity axiom, to obtain a weakening of the impossibility result that holds for $m>p\geq 2$, with the existence of an essential dictator.\footnote{If for every classification problem, objects are always classified according to a permutation of the classification of a given individual, it is essentially a dictatorship.}

The intention of this paper is to follow the line of \citet{maniquet2016theorem}  and \citet{cailloux2023classification}, who study counterparts for the classification aggregation problem of classical results in standard preference aggregation (with \citeauthor{Arrow1951-ARRIVA}'s \citeyearpar{Arrow1951-ARRIVA} and \citeauthor{WILSON1972478}'s \citeyearpar{WILSON1972478} theorems, respectively).
In this work, we draw inspiration from \citeauthor{sen1970paretianliberal}'s \citeyearpar{sen1970paretianliberal} impossibility result of a Paretian liberal and adapt various versions of the liberal axiom to our context. 
\citeauthor{sen1970paretianliberal}'s \citeyearpar{sen1970paretianliberal} shows that even a minimal level of individual rights (where people have control over personal choices) can conflict with Pareto efficiency (the idea that if everyone agrees something is better, we should choose that outcome).
In our setting, objects do not necessarily belong to individuals' private spheres; hence, instead of referring to it as the liberal axiom, we term it the expertise axiom.
This property requires that a collective decision process should ensure that each individual with expertise can unilaterally determine the classification of a specific object or the objects assigned into a particular category. 
The results obtained here are significant in cases where individuals have a direct connection to some of the objects they need to classify or the categories involved in the classification process.
Consider, for instance, a group of editors responsible for overseeing the publication of a `Handbook in Economics', where each major area of the discipline must be represented by a dedicated chapter. 
Some editors may submit papers for specific chapters themselves, while others may have specialised expertise in particular fields of economics. 
In this context, two reasonable approaches could be proposed: allowing editors who are also authors to decide the appropriate chapter for their own submissions or assigning the task to editors with expertise in the corresponding chapters.
% For example, consider an economics conference organizing committee tasked with classifying all submitted papers into appropriate sessions, where every session must have at least one paper. 
% It is likely that some committee members have submitted papers themselves, while others possess expertise in specific areas of economics. 
% Given this, one could propose two reasonable approaches: allowing authors to decide which session their paper belongs in or having specialists assign papers to the sessions they are experts on.
% For example, consider a group of people, some of whom own houses in a particular neighbourhood while others are prospective buyers. 
% They collectively need to decide on the colour each house should be painted in and there is a set of mandatory colours that all should be used. 
% From a liberal perspective, this decision should, at a minimum, respect the opinions of the homeowners regarding their own houses.
% Moreover, if some of the individuals also own the paint that will be used, we can require that the usage of that paint respects at least some of the owner's opinion.
These ideas can be captured by an expertise axiom, indicating different ways in which an individual can be decisive over the classification of an object or into a category.
We show that in general, the expertise axiom is incompatible with natural properties such as Unanimity or Independence.
Furthermore, some versions of expertise alone are not even satisfiable in this setting, where no category can be left empty.
This represents a departure from results in the Group Identification Problem, where there are aggregators where the individuals are decisive for certain classifications \citep{miller2008group,fioravanti2021alternative}.
The key distinction between these two settings is that in the Group Identification Problem, it is not necessary that all the categories get an object classified into.

The implications of the impossibility results presented in this paper extend to several domains where classification tasks are fundamental. 
For instance, in organisational decision-making, where tasks or resources must be allocated to individuals or groups, the expertise axiom reflects the need to respect individual preferences based on their expertise. 
The results here imply that in striving to ensure that no single individual's preferences dominate across all tasks, we may encounter situations where it is impossible to maintain other desirable properties such as Unanimity or Independence. 
This finding could be particularly relevant in settings where collaborative decision-making processes are used, such as in project management or committee-based allocations.
Similarly, in the field of multi-agent systems, where multiple autonomous agents must classify resources or tasks into categories such as `high priority', `medium priority', and `low priority', the expertise axiom might reflect the need for each agent to have control over at least one task or resource. 
However, the impossibility results indicate that achieving a fair and efficient aggregation of these classifications might be inherently challenging when trying to balance `control', Unanimity, and Independence. 
This has practical implications for the design of decentralised decision-making algorithms, used, for example, in distributed computing or collaborative robotics.

The plan of the paper is as follows.
Section \ref{basic} presents the basic notions and axioms that we use, while we present the results in Section \ref{results}.
Finally, Section \ref{remarks} contains some concluding remarks.

\section{Basic Notions and Axioms}\label{basic}
Let $N=\{1,\ldots,n\}$, with $n\geq 2$, be a finite set of individuals and let $X=\{x_1,\ldots,x_m\}$ be a set of $m$ objects that need to be classified into the $p$ categories of a set $P$, with $p\geq 2$.
In this setting, introduced by \citet{maniquet2016theorem}, \emph{classifications} are surjective mappings $c:X\rightarrow P$, that is, every category must have at least one object classified into.
Thus we have that $m\geq p\geq 2$.
We use $\C$ to denote the set of classifications, and every \mbox{$\B{c}=(c_1,\ldots,c_n)\in \C^N$} is a {\it classification aggregation problem}, with $c_i$ indicating the classification given by individual $i$. 
A {\it classification aggregation function} (CAF) is a mapping \mbox{$\alpha:\C^N\rightarrow \C$} such that $\alpha(\B{c})(x)$ indicates the category where object $x$ is classified into.
We call the outcome of $\alpha$ the {\it social classification}.
Given a category $t\in P$, we denote the inverse image of $t$ as $\alpha(\B{c})^{-1}(t)$. 
Formally, ${\alpha(\B{c})}^{-1}(t)=\{x\in X\mid \alpha(\B{c})(x)=t\}$.

In the following, we introduce a number of axioms, i.e., fundamental normative requirements that, in this specific classification process, we consider a reasonable CAF should satisfy.
The first property states that if there is an object that is unanimously classified by the individuals, then the CAF has to classify that object accordingly.
This property is the unary interpretation of \citeauthor{Arrow1951-ARRIVA}'s \citeyearpar{Arrow1951-ARRIVA} and \citeauthor{sen1970paretianliberal}'s \citeyearpar{sen1970paretianliberal} Pareto principle, introduced by \citet{maniquet2016theorem}.
\begin{defi}[Unanimity]
A CAF is {\em unanimous} if for all $\B{c}\in\C^N$, all $x\in X$, and all $t\in P$ such that $c_1(x)=\cdots=c_n(x)=t$, it is the case that $\alpha(\B{c})(x)=t$.    
\end{defi}
Next, we introduce different interpretations of how an individual can be decisive in a classification problem.
Individuals can have decision power over objects, categories, or pairs of objects and categories.
We say that individual $i$ is {\em decisive over object} $x$, if for all $\B{c}\in\C^N$, and all $t\in P$, it is the case that $c_i(x)=t$ implies $\alpha(\B{c})(x)=t$.\footnote{Although this definition is expressed as an implication for presentation purposes, it is equivalent to using a bi-conditional, given that an individual must classify every object.}
An individual $i$ is {\em categorically decisive over category} $t$, if for all $\B{c}\in\C^N$, and all $x\in X$, it is the case that $\alpha(\B{c})(x)=t$ implies  $c_i(x)=t$.
Finally, individual $i$ is {\em minimally decisive over the object $x$ and category $t$}, if for all $\B{c}\in\C^N$, it is the case that $c_i(x)=t$ if, and only if, $\alpha(\B{c})(x)=t$.\footnote{These formal definitions provide different interpretations of how an individual can be decisive in this setting, which are useful for identifying impossibility results. 
While other definitions, such as those presented in Remark 1, are of interest, they fall outside the scope of this work, which focuses on illustrating the complexity of aggregation processes with decisive individuals.}
Now we present three axioms requiring the existence of two decisive individuals, an adaptation to this setting of \citeauthor{sen1970paretianliberal}'s \citeyearpar{sen1970paretianliberal} minimally liberal principle.
The goal is to ensure a minimum number of decisive individuals, sufficient to prevent the emergence of a dictator while remaining practically demanding.\footnote{The objects of Definition~\ref{exp} and Definition~\ref{minexp} must be different, otherwise an object might end up being classified into two different categories (a contradiction to the definition of a classification function).}
\begin{defi}[Expertise]\label{exp}
A CAF is \emph{expert} if there exist two individuals $i,j\in N$ and two different objects $x,y\in X$ such that for all $\B{c}\in\C^N$, individuals $i$ and $j$ are decisive over $x$ and $y$, respectively.  
\end{defi}
\begin{defi}[Categorical Expertise]
A CAF is \emph{categorically expert} if there exist two individuals $i,j\in N$ and two categories $t,t'\in P$ such that for all $\B{c}\in\C^N$, individuals $i$ and $j$ are categorically decisive over $t$ and $t'$, respectively. 
\end{defi}
\begin{defi}[Minimal Expertise]\label{minexp}
A CAF is {\em minimally expert} if there exist two individuals $i,j\in N$ and two pairs $(x,t),(y,t')\in X\times P$, with $x\neq y$, such that for all $\B{c}\in\C^N$, individuals $i$ and $j$ are minimally decisive over $(x,t)$ and $(y,t')$, respectively.   
\end{defi}
It is easy to see that Minimal Expertise is weaker than Expertise, with the former being implied by the latter.
Finally, we present a unary interpretation of \citeauthor{Arrow1951-ARRIVA}'s \citeyearpar{Arrow1951-ARRIVA} Independence of Irrelevant Alternatives, introduced by \citet{maniquet2016theorem}, which states that the social classification of an object in two different profiles is the same if the individual's classifications of that object are the same in both profiles. 
\begin{defi}[Independence]
A CAF is {\em independent} if given $\B{c},\B{c'}\in\C^N$ and $x\in X$ such that \mbox{$c_i(x)=c'_i(x)$} for all $i\in N$, it is the case that \mbox{$\alpha(\boldsymbol{c})(x)=\alpha(\boldsymbol{c'})(x)$}.    
\end{defi}

\section{Results}\label{results}
We start by looking at the implications of Minimal Expertise, and its conjunction with Unanimity.
Our first result resembles \citeauthor{sen1970paretianliberal}'s \citeyearpar{sen1970paretianliberal} impossibility of a Paretian liberal, highlighting the conflict that arises between the Pareto principle and individual expertise (individuals being decisive over objects or categories in which they are experts).
\begin{teo}\label{unaminlib}
 For $m\geq p\geq 2$, there is no CAF that satisfies Unanimity and Minimal Expertise.  
\end{teo}
\begin{proof}
Let $\alpha$ be a CAF that satisfies Unanimity and Minimal Expertise, and assume, without loss of generality, that the individual $1$ is decisive over the pair $(x,t')$ and the individual $2$ is decisive over the pair $(y,t'')$.
Now suppose that $t'\neq t''$, $t'=t_1$, $t''=t_2$, and consider the following classification aggregation problem $\B{c}$:\footnote{The rows in the table indicate the objects that the individuals classify into a given category, while the columns are the individual's classifications.}
\begin{center}
\begin{tabular}{c|c c c}
$\B{c}$    & $1$ & $2$ & $N\setminus\{1,2\}$  \\
    \hline
 $t_1$    & $y$ & $x$ &$x$  \\
 $t_2$    & $\{x\}\cup A_2$ & $\{y\}\cup A_2$ &$\{y\}\cup A_2$  \\
$t_3$    & $A_3$ & $A_3$ &$A_3$  \\
$\vdots$& $\vdots$ & $\vdots$ &$\vdots$  \\
$t_p$    & $A_p$ & $A_p$ &$A_p$  
\end{tabular}
\end{center}
where $A_i\subset X$, $A_i\neq \emptyset$ for $i=\{3,\ldots,p\}$, for all $i$ and $j$ we have that $A_i\cap A_j=\emptyset$, and that $\bigcup_{i=2}^pA_i=X\setminus\{x,y\}$.
So, the classification given by all the individuals is the same for every object except for $x$ and $y$.
Thus, by Unanimity, $A_i\subseteq\alpha(\B{c})^{-1}(t_i)$ for $i\in\{3,\ldots,p\}$.
In particular, by Unanimity and Minimal Expertise, we have that $\{y\}\cup A_2\subseteq\alpha(\B{c})^{-1}(t_2)$ and $\alpha(\B{c})(x)\neq t_1 $.
Thus $\alpha(\B{c})^{-1}(t_1)=\emptyset$, concluding that $\alpha(\B{c})$ is not a classification function.
The proof is similar for the case where $t'= t''$. 
If we assume $t'=t''=t_1$, then using the same classification aggregation problem $\B{c}$, we have that $x,y\notin \alpha(\B{c})^{-1}(t_1)$, and thus $\alpha(\B{c})^{-1}(t_1)=\emptyset$.
\end{proof}
\begin{rem}\label{minimallysemidecisive}\em If we consider that an individual $i$ can be {\em minimally semi-decisive for a pair $(x,t)$}, that is, if $c_i(x)=t$ implies $\alpha(\B{c})(x)=t$, then there are CAF's that are unanimous and can have two or more minimally semi-decisive individuals.\footnote{Note that in \citet{sen1970paretianliberal}, requiring individuals to be semi-decisive still leads to an impossibility result.}
One example is the following CAF.
Let $x_1\succ x_2\succ \ldots\succ x_m$ and $t_1\succ t_2\succ\ldots\succ t_p$ be given orders for the objects and the categories, respectively, and assume $1$ and $2$ are minimally semi-decisive over the pairs $(x_1,t_1)$ and $(x_2,t_2)$, respectively. 
For the cases where $c_1(x_1)=t_1$ or $c_2(x_2)=t_2$, the CAF classifies those objects accordingly. 
Then, classifies all the unanimous classifications accordingly.
And finally, it classifies the first unassigned object to the first empty category, the second unassigned object to the second empty category, and so on, following the given orders, until all the objects are classified.
For the case where $m>p$, when all the categories have one object, unassigned objects are classified into the categories following the given order. 
If neither $c_1(x_1)=t_1$ nor $c_2(x_2)=t_2$, the CAF skips the first step.
This rule also works for the case where $1$ and $2$ are minimally semi-decisive over $(x_1,t)$ and $(x_2,t)$, when $m>p$.
For $m=p$, there is no CAF for this case, as one category might remain without an object being classified into. 

\end{rem}

For the particular case of $2$ objects and $2$ categories, even the sole requirement of the existence of minimally decisive individuals is excessively demanding.
\begin{propi}\label{m=p=2}
 For $m=p= 2$, there is no CAF that satisfies Minimal Expertise.   
\end{propi}
\begin{proof}
Let $\alpha$ be a CAF that satisfies Minimal Expertise, and assume, without loss of generality, that the individual $1$ is decisive over the pair $(x,t')$ and the individual $2$ is decisive over the pair $(y,t'')$.
Now suppose that $t'\neq t''$, $t'=t_1$, $t''=t_2$ and consider the following classification aggregation problem $\B{c}$:
\begin{center}
\begin{tabular}{c|c c c}
  $\B{c}$  & $1$ & $2$ & $N\setminus\{1,2\}$  \\
    \hline
 $t_1$    & $y$ & $x$ &$x$  \\
 $t_2$    & $x$ & $y$ &$y$  

\end{tabular}
\end{center}
By Minimal Expertise, we have that $y\in\alpha(\B{c})^{-1}(t_2)$ and $\alpha(\B{c})(x)\neq t_1 $.
Thus $\alpha(\B{c})^{-1}(t_1)=\emptyset$, concluding that $\alpha(\B{c})$ is not a classification function.
The proof is similar for the case where $t'= t''$ (we can use the previous classification aggregation problem $\B{c}$).
\end{proof}
\begin{rem}\label{minimallyliberalpossibility}\em There exist minimally expert CAFs for $m\geq p> 2$.
Let individuals $1$ and $2$ be minimally decisive over $(x,t')$ and $(y,t'')$, respectively, with $t'\neq t''$.
A CAF that for the cases where $c_1(x_1)=t'$ or $c_2(x_2)=t''$, classifies those objects accordingly and assigns the remaining objects so no category is left empty, and otherwise, classifies the objects according to a pre-given classification, satisfies Minimal Expertise. 
When $t'=t''$, we need that $m>p>2$ (otherwise there might be categories that are left empty).
\end{rem}
Next, we demonstrate that requiring a CAF to be both independent and minimally expert is quite demanding, though not as stringent as requiring the aggregator to be unanimous.
\begin{propi}\label{indeplib}
Let $p+1\geq m\geq p$. 
There is no CAF that satisfies Minimal Expertise and Independence.    
\end{propi}
\begin{proof}
Let $\alpha$ be a CAF that satisfies Independence and Minimal Expertise, and assume, without loss of generality, that individual $1$ is minimally decisive over the pair $(x,t')$ and individual $2$ is minimally decisive over the pair $(y,t'')$.
If $m=p$, every classification is such that there is exactly one object in every category, and if $m=p+1$, every classification is such that there is exactly one object in every category but one, that has two objects classified into.
Let $t'\neq t''$, $t'=t_1$, $t'=t_2$, and consider the following classification aggregation problem $\B{c}$:
\begin{center}
\begin{tabular}{c|c c c}
 $\B{c}$   & $1$ & $2$ & $N\setminus\{1,2\}$  \\
    \hline
 $t_1$    & $x$ & $y$ &$x$  \\
 $t_2$    & $y$ & $x$ &$y$  \\
$\vdots$& $\vdots$ & $\vdots$ &$\vdots$ 
\end{tabular}
\end{center}
such that the rest of the table is completed to be a classification aggregation problem.
Thus, by Minimal Expertise, we have that $x\in \alpha(\B{c})^{-1}(t_1)$ and $y\notin\alpha(\B{c})^{-1}(t_2)$.
Assume that $z\in \alpha(\B{c})^{-1}(t_2)$, with $z\neq y$.
% Assume, without loss of generality, that $\alpha(\B{c})(z)=t_1$.\footnote{If $z$ is classified into a different category, for example $t'$, we can construct a problem $\B{c}^{\ast}$ such that individuals $1$ and $2$ classify $x$ and $y$ into $t'$, and obtain the same contradiction. }
Now consider the following classification aggregation problem $\B{c}'$:
\begin{center}
\begin{tabular}{c|c c c}
  $\B{c}'$  & $1$ & $2$ & $N\setminus\{1,2\}$  \\
\hline
 $t_1$    & $x$ & $x$ &$x$  \\
 $t_2$    & $y$ & $y$ &$y$  \\
$\vdots$& $\vdots$ & $\vdots$ &$\vdots$ 
\end{tabular}
\end{center}
such that the rest of the table is completed to be a classification aggregation problem, and the rows from $t_3$ to $t_p$ are the same as in $\B{c}$.
Then, by Independence and Minimal Expertise, we have that $\{z,y\}= \alpha(\B{c}')^{-1}(t_2)$, and thus, $\{x\}=\alpha(\B{c}')^{-1}(t_1)$.
Otherwise, there are two categories with two objects classified into.
For the case where $m=p$, this already leads to a contradiction.
Finally, consider the classification aggregation problem $\B{c}''$:
\begin{center}
\begin{tabular}{c|c c c}
  $\B{c}''$  & $1$ & $2$ & $N\setminus\{1,2\}$  \\
\hline
 $t_1$    & $y$ & $x$ &$x$  \\
 $t_2$    & $x$ & $y$ &$y$  \\
$\vdots$& $\vdots$ & $\vdots$ &$\vdots$ 
\end{tabular}
\end{center}
such that the rest of the table is completed to be a classification aggregation problem, and the rows from $t_3$ to $t_p$ are the same as in $\B{c}'$.
By Independence and Minimal Expertise, we have that $\{z,y\}= \alpha(\B{c}'')^{-1}(t_2)$, $x\notin\alpha(\B{c}'')^{-1}(t_1)$ and $\alpha(\B{c}'')(w)=\alpha(\B{c}')(w)\neq t_1$ for all $w\in X\setminus\{x,y,z\}$.
Hence the category $t_1$ is left empty and $\alpha(\B{c}'')$ is not a classification function.

For $t'=t''$, we can assume that $t'=t_1$.
Thus, we need that $m>p$ and we are in the case $m=p+1$.
Consider the classification aggregation problems $\B{c}$ and $\B{c''}$ previously introduced.
By Minimal expertise, $x,y\in\alpha(\B{c})^{-1}(t_1)$.
Then, it must be that there is exactly one object in each category $t_i\neq t_1$.
Let $z_i\in X\setminus\{x,y\}$ and $\{z_i\}=\alpha(\B{c})^{-1}(t_i)$, for $i\in\{2,\ldots,p\}$.
By Independence, we have that $z_i\in\alpha(\B{c''})^{-1}(t_i)$, for $i\in\{2,\ldots,p\}$, and by Minimally Expertise we have that $x,y\notin \alpha(\B{c''})^{-1}(t_1)$.
So $\alpha(\B{c''})^{-1}(t_1)=\emptyset$, and thus $\alpha(\B{c''})$ is not a classification function.

%Thus there are three objects classified into $t_1$, so $\alpha(\B{c}')$ is not a classification function.
\end{proof}
In light of the Group Identification Problem, this result might be seen as surprising, as there are independent and expert rules.
An example is the Liberal aggregator, where every individual is decisive over their own classification \citep[see, for more examples,][]{fioravanti2021alternative}.
\begin{rem}\label{liberalindependentpossibility}\em There are independent and expert (thus minimally expert) CAF's for $m\geq p+2$.
For example, assume $1$ is decisive over $x_1$ and $2$ is decisive over $x_2$, and consider the CAF that fixes all the objects in $X\setminus\{x_1,x_2\}$ to $p$ different categories, and then assigns $x_1$ and $x_2$ according to individuals $1$ and $2$ classifications.
Moreover, if $m\geq p+n$, where $n$ is the number of individuals, we can have a rule such that all the individuals are decisive over an object, and under the same reasoning of the previous example, obtain an independent and expert CAF. 
\end{rem}
Now we turn our attention to the analysis of Expertise alone.
We have already shown that the conjunction of either Unanimity or Independence and Minimal Expertise is rather demanding, so it is expected that at least for some configurations of the number of objects and categories, just Expertise alone is prohibitive.
That is what the next result shows.
\begin{propi}\label{lib}
Let $m=p$. There is no CAF that satisfies Expertise.    
\end{propi}
\begin{proof}
If $m=p$, then there is exactly one object in each category.
Let $\alpha$ be a CAF that satisfies Expertise, and assume, without loss of generality, that individual $1$ is decisive over object $x$ and individual $2$ is decisive over object $y$.
Consider the following classification aggregation problem $\B{c}$:
\begin{center}
\begin{tabular}{c|c c c}
 $\B{c}$   & $1$ & $2$ & $N\setminus\{1,2\}$  \\
    \hline
 $t_1$    & $y$ & $x$ &$x$  \\
 $t_2$    & $x$ & $y$ &$y$  \\
$\vdots$& $\vdots$ & $\vdots$ &$\vdots$ 
\end{tabular}
\end{center}
such that the rest of the table is completed to be a classification aggregation problem.
Then we have that $x$ and $y$ are both classified into $t_2$, so $\alpha(\B{c})$ is not a classification function.
\end{proof}%We still obtain an impossibility result for $m=p+1$, with the additional requirement of Independence.
\begin{rem}\label{liberalpossibility}\em There are expert CAF's for $m=p+1$.
The CAF that classifies $x$ and $y$ according to the classifications given by the individuals that are decisive over them, and then classifies the rest of the objects following a given order such that no category is left empty, is an expert CAF.
% The CAF presented in Remark~\ref{minimallysemidecisive}, but skips the second step (classifying the unanimous classifications), satisfies Expertise. 
\end{rem}
\begin{rem}\label{dictator}\em There is a CAF that satisfies Unanimity and Independence, but not Minimal Expertise.
It is the dictatorial CAF, such that for all $\B{c}\in\mathcal{C}^N$, it is the case that $\alpha(\mathbf{c})= c_i$, for a fixed $i\in N$.
Moreover, it is the unique CAF that satisfies those two properties \citep{maniquet2016theorem}.
    
\end{rem}
Our final result shows that requiring individuals to have decisive power over the categories turns out to be prohibitive, even if we do not impose extra normative requirements.
\begin{propi}\label{catliberal}
There is no CAF that satisfies Categorical Expertise.    
\end{propi}
\begin{proof}
Let $\alpha$ be a CAF that satisfies Categorical Expertise, and assume, without loss of generality, that the individual $1$ is decisive over the category $t'$ and the individual $2$ is decisive over the category $t''$.
Let $t'\neq t''$, $t'=t_1$, $t'=t_2$, and consider the following classification aggregation problem $\B{c}$:
\begin{center}
\begin{tabular}{c|c c c}
 $\B{c}$   & $1$ & $2$ & $N\setminus\{1,2\}$  \\
    \hline
 $t_1$    & $x$ & $y$ &$x$  \\
 $t_2$    & $y$ & $x$ &$y$  \\

$\vdots$& $\vdots$ & $\vdots$ &$\vdots$ 
\end{tabular}
\end{center}
such that the rest of the table is completed to be a classification aggregation problem.
Then, by Categorical Expertise, we have that $\emptyset\neq\alpha(\B{c})^{-1}(t_1)\subseteq \{x\}$ and $\emptyset\neq\alpha(\B{c})^{-1}(t_2)\subseteq \{x\}$, leading to a contradiction.

For the case $t'=t''$, we can assume that $t'=t_1$ and consider the following classification aggregation problem $\B{c}'$:
\begin{center}
\begin{tabular}{c|c c c}
 $\B{c}'$   & $1$ & $2$ & $N\setminus\{1,2\}$  \\
    \hline
 $t_1$    & $x$ & $y$ &$y$  \\
 $t_2$    & $y$ & $x$ &$x$  \\

$\vdots$& $\vdots$ & $\vdots$ &$\vdots$ 
\end{tabular}
\end{center}
such that the rest of the table is completed to be a classification aggregation problem.
Then, by Categorical Expertise, we have that $\emptyset\neq\alpha(\B{c})^{-1}(t_1)\subseteq \{x\}$ and $\emptyset\neq\alpha(\B{c})^{-1}(t_1)\subseteq \{y\}$, leading to a contradiction.
\end{proof}
If we consider this result in light of the Group Identification Problem, it strikes us as surprising.
When the requirement of not leaving a category empty is not imposed, it is easy to think of a rule that satisfies Categorical Expertise.
For example, for the case of only two categories, where objects are the same individuals \citep[the `Who is a J?' original setting,][]{kasher1997question}, the aggregator that classifies an individual into category J if, and only if, individuals $1$ and $2$ classify them as $J$, and not in $J$ otherwise, is categorically expert.

\section{Final Remarks}\label{remarks}

In this paper, we address the problem of classifying $m$ objects into $p$ different categories, ensuring that no category remains empty and recognizing that certain individuals, deemed experts, have decisive power over objects or categories. 
Even with the minimal requirement that only two individuals possess decisiveness, the potential for aggregators emerges only under specific configurations of the number of objects and categories.
We derive a result analogous to \citeauthor{sen1970paretianliberal}'s \citeyearpar{sen1970paretianliberal} impossibility theorem of a Paretian liberal, demonstrating that there is no Classification Aggregation Function that satisfies both Unanimity and Minimal Expertise, that is when individuals have decisive power over pairs of objects and categories.
This finding illustrates how the Pareto principle remains in conflict with individual expertise within this context.
The best we can do is to show the existence of independent and expert CAFs, for the cases where there are at least two more objects than categories, although for particular values $(m=p)$ even Expertise alone cannot be satisfied.
We also show that it is not possible in this setting to require individuals to be decisive over categories, even if we do not impose additional normative requirements.
A summary of our results can be found in Table~\ref{summary}.

\begin{table}[t]
\begin{center}
\centerline{\begin{tabular}{l|c|c|c}
   & Without additional axioms & Unanimity & Independence  \\ \hline
M. Expertise & No, for $m=p=2$ (Prop.~\ref{m=p=2})  & No (Thm.~\ref{unaminlib}) & No, for $p+1\geq m\geq p$ (Prop.~\ref{indeplib})  \\
 & Yes, for $m\geq p>2$ (Rem.~\ref{minimallyliberalpossibility}) & & Yes, for $m\geq p+2$ (Rem.~\ref{liberalindependentpossibility}) \\ \hline
Expertise & No, for $m=p$ (Prop.~\ref{lib})  & No (Thm.~\ref{unaminlib}) & No, for $p+1\geq m\geq p$ (Prop.~\ref{indeplib})  \\
 & Yes, for $m\geq p+1$ (Rem.~\ref{liberalindependentpossibility} and \ref{liberalpossibility}) & & Yes, for $m\geq p+2$ (Rem.~\ref{liberalindependentpossibility}) \\
 \hline
C. Expertise & No (Prop.~\ref{catliberal})   & No (Prop.~\ref{catliberal}) & No (Prop.~\ref{catliberal})  \\ 
\end{tabular}}
\end{center}
\caption{Summary of our results.}
\label{summary}
\end{table}

In essence, the expertise axiom is a powerful tool for ensuring that decision-making processes are not dominated by a single perspective. 
However, its application necessitates careful consideration of the trade-offs involved, particularly when other principles like fairness, consistency, and independence are also valued. 
The results in this paper highlight the complexities that arise when attempting to create classification systems that are both fair and responsive to a diversity of decision-makers, offering critical insights into the design of such systems across various domains.

\bibliography{ref}
\end{document}